\documentclass[a4paper,12pt]{amsart}
\usepackage[T1]{fontenc}
\usepackage[utf8]{inputenc}
\usepackage[margin=2.5cm]{geometry}

\usepackage{amsmath, amssymb,amsfonts}
\usepackage{thm-restate}
\usepackage{array}
\usepackage{multirow}
\usepackage[dvipsnames]{xcolor}
\usepackage{tikz}

\usepackage{hyperref}
\hypersetup{
  pdftitle = {Renaming in distributed certification},
  pdfauthor = {N. Bousquet, L. Esperet, L. Feuilloley, and S. Zeitoun},
  colorlinks = true,
  linkcolor = black!30!red,
  citecolor = black!30!green
}
  \usepackage[capitalise, compress, nameinlink, noabbrev]{cleveref}
\crefname{rule}{}{}
\creflabelformat{rule}{#2(R#1)#3}

\usepackage{paralist}

\newtheorem{theorem}{Theorem}[section]
\newtheorem{corollary}[theorem]{Corollary}
\newtheorem{definition}[theorem]{Definition}

\newtheorem{lemma}[theorem]{Lemma}

\newtheorem{question}{Question}

\def\cqedsymbol{\ifmmode$\lrcorner$\else{\unskip\nobreak\hfil
\penalty50\hskip1em\null\nobreak\hfil$\lrcorner$
\parfillskip=0pt\finalhyphendemerits=0\endgraf}\fi}

\newcommand{\N}{\mathbb{N}}

\let\le\leqslant
\let\ge\geqslant
\let\leq\leqslant

\newcommand{\NN}{\mathbb{N}} % integers
\newcommand{\T}{\mathcal{T}}
\newcommand{\G}{\mathcal{G}}
\newcommand{\renaming}{\mathsf{Renaming}}
\newcommand{\idroot}{\mathsf{Root}}
\newcommand{\idparent}{\mathsf{Parent}}
\newcommand{\distance}{\mathsf{Distance}}
\newcommand{\subtree}{\mathsf{Subtree}}

\title{Renaming in distributed certification}

\author[N. Bousquet]{Nicolas Bousquet}
\address[N.~Bousquet]{CNRS, INSA Lyon, UCBL, LIRIS, UMR5205, F-69622 Villeurbanne, France}
\email{nicolas.bousquet@cnrs.fr}

\author[L.~Esperet]{Louis Esperet}
\address[L.~Esperet]{Université Grenoble Alpes, CNRS, Laboratoire G-SCOP,
  Grenoble, France}
\email{louis.esperet@grenoble-inp.fr}

\author[L.~Feuilloley]{Laurent Feuilloley}
\address[L.~Feuilloley]{CNRS, INSA Lyon, UCBL, LIRIS, UMR5205, F-69622 Villeurbanne, France}
\email{laurent.feuilloley@cnrs.fr}

\author[S.~Zeitoun]{Sébastien Zeitoun}
\address[S.~Zeitoun]{CNRS, INSA Lyon, UCBL, LIRIS, UMR5205, F-69622 Villeurbanne, France}
\email{sebastien.zeitoun@univ-lyon1.fr}

\thanks{The introduction and the part about global certification heavily rely on a brief announcement that appeared at PODC 2024, by a subset of the authors \cite{BousquetFZ24}. The authors are partially supported by the French ANR Projects TWIN-WIDTH
  (ANR-21-CE48-0014-01) and ENEDISC (ANR-24-CE48-7768), and by LabEx
  PERSYVAL-lab (ANR-11-LABX-0025).}

\begin{document}
\date{}

\begin{abstract}
Local certification is the area of distributed network computing asking the following question: How to certify to the nodes of a network that a global property holds, if they are limited to a local verification? 

In this area, it is often essential to have identifiers, that is, unique integers assigned to the nodes. In this short paper, we show how to reduce the range of the identifiers, in various settings. More precisely, we show how to rename identifiers in the classical local certification setting, when we can (resp.\ cannot) choose the new identifiers, and we show how a global certificate can help to encode very compactly a new identifier assignment that is not injective in general, but still useful in applications. 

We conclude with a number of  applications of these results: For every constant $\ell$, there are  local certification schemes for the properties of having clique number at most $\ell$,  having diameter at most $\ell$, and having independence number at most~2, with certificates of size $O(n)$. We also show that there is a global certification scheme for bipartiteness with a certificate of size $O(n)$. All these results are optimal.
\end{abstract}

\maketitle

%%%%%%%%%%%%%%%%%%%%%
\section{Introduction}
%%%%%%%%%%%%%%%%%%%%%

The topic of distributed certification originates from self-stabilization in distributed computing, where the nodes of a network are provided with some pieces of information called \emph{certificates}. These certificates can either be local (each node receives its own certificate), or global (there is a unique certificate, which is the same for all the nodes). The aim of the nodes is then to decide if the network satisfies a given property. To do so, each node should take its decision (accept or reject) based only on its local view in the network, which consists in its neighbors and their certificates.
The correctness requirement for a certification scheme is the following: for every network, the property is satisfied if, and only if, there exists an assignment of the certificates such that all the nodes accept.
Unsurprisingly, the parameter we want to optimize is the size of the certificates, which is usually expressed as a function of $n$, the number of nodes in the network.  
We refer to the survey~\cite{Feuilloley21} for an introduction to local certification.

\subsection*{Renaming in local certification}

In almost all the literature on distributed certification, one assumes that the nodes are equipped with unique identifiers, that is, every node is given a different integer, that it uses as a name. The classical assumption in distributed graph algorithms is that the range of these identifiers is $[n^c]:=\{1, \ldots, n^c\}$ for some constant $c>1$, where $n$ is the number of nodes. This is equivalent to saying that identifiers are binary words of length $O(\log n)$. One reason one does not use integers from $[n]$, is that it gives too much power to the algorithm, for example electing a leader can be done with no communication by electing the vertex with identifier 1.
The range $[n^c]$ is also the classical assumption in distributed certification, and having directly range $[n]$ allows to easily shave log factors from known results, as we will see at the end of the paper. 

Therefore, in general, a natural question is whether we can perform a renaming procedure efficiently, that is, starting from identifiers in some arbitrary range $[M]$, and efficiently getting to range $[n]$ (which is equivalent to have an explicit bijection from the vertex set to $[n]$). This type of task is usually studied in the context of shared-memory systems (see \emph{e.g.}~\cite{Alistarh15}), but it is also meaningful in synchronous systems. 

In the context of certification, this translates to the question: suppose every node is given as input a new integer from the smaller range, can we certify that these integers form a proper identifier assignment? We prove two results related to this question.

\begin{restatable}{theorem}{ThmSeb}
    \label{thm:seb}
    In any $n$-vertex graph, there exists an identifier renaming from range $[M]$ to range $[n]$ which can be locally certified with $O(\log M)$ bits. 
\end{restatable}

A direct consequence of this theorem is that for local certification with certificates of size $\Omega(\log n)$, assuming identifiers in $[n^c]$ or $[n]$ is equivalent. Indeed, when nodes are given identifiers in $[n^c]$, one can encode in the certificates a new identifier assignment from $[n]$ as well as its certification with $O(\log n)$ bits. 

In Theorem~\ref{thm:seb}, the new identifier assignment has a specific form that is dictated by the structure of the graph. In some cases, we cannot choose this new identifier assignment, as it might be given by an adversary. We prove that such a renaming can still be certified compactly, albeit with significantly larger certificates.

\begin{restatable}{theorem}{ThmLouis}\label{thm:ThmLouis}
    In any $n$-vertex graph, any identifier renaming from range $[M]$ to range $[n]$ can be locally certified with $O(n + \log M)$ bits. 
\end{restatable}

This theorem is a generalization of a similar result restricted to  paths and identifier range~$[n]$ that appeared in~\cite{KormanKP10}. The same paper proves a lower bound of $\Omega(n)$.  Note that the setting of Theorem \ref{thm:ThmLouis} is equivalent to the \textsc{Permutation} problem, where the nodes of the network are given integers from $[n]$ and they need to certify that these values form a permutation of $[n]$ (see \cite{NPY20} for more on this problem, in the setting of interactive certification protocols). We will see several simple applications of Theorems \ref{thm:seb} and \ref{thm:ThmLouis} in Section \ref{sec:applications}: for every $\ell$, there are  local certification schemes for $K_\ell$-freeness (the property of not containing a clique of size $\ell$) and diameter at most $\ell$, with certificates of size $O(n)$. These results are optimal, and improve earlier results by a logarithmic factor. We will also prove a similar result for the property of having independence number at most~2.

We will also show that Theorem~\ref{thm:ThmLouis} can be combined with hashing techniques described in the next paragraph to provide an efficient certification protocol for the \textsc{Distinctness} problem (see Section~\ref{sec:applications}). In this problem, also studied in \cite{NPY20}, the nodes are given values from $[m]$ and they have to certify that the values are pairwise distinct. \textsc{Permutation} and \textsc{Distinctness} are useful in the context of isomorphism or non-isomorphism problems~\cite{NPY20}.
This efficient certification protocol for \textsc{Distinctness} will also enable us to point out a link between distributed certification and universal graphs (namely, how a universal graph for some monotone class can be used to obtain a certification protocol for this class).

\subsection*{Global certification and perfect hashing}

As mentioned above, there are two natural levels of locality in certification. In the first case, the certificates are local, and the verification is local too; in the second case, the certificates are global, but the verification remains local. When speaking about local or global certification, we thus refer to the locality or globality of the certificates (and not of the verification, which is always local). In general, these two levels of certification are related, because bounds for one can be derived from bounds for the other. Namely, a global certification scheme is a particular case of a local one, and conversely, a local certification scheme can be transformed into a global one by giving as global certificate the list of the local certificates of each node in the network (so that each node can simulate the local certification scheme by recovering its own local certificate from the global one, see~\cite{FeuilloleyH18} for more details).
However, these generic transformations are usually not optimal.

The reason we mention the relation between global and local certification, is that identifiers play a significant role there. Consider the example of certifying that a graph is bipartite. This can be done in  local certification by assigning to every node a single bit indicating its part in the bipartition. As observed in \cite{FeuilloleyH18}, we can translate this into a global certificate of size $\Theta(n\log n)$ consisting of all the pairs (identifiers, bit). Authors in~\cite{FeuilloleyH18} conjectured that this is optimal in the standard case where the range of identifiers is polynomial in $n$ (see also Open Problem 9 in~\cite{Feuilloley21}). In other words, they conjectured that the answer to the following question is positive:

\begin{restatable}{question}{ConjectureGlobalBipartiteness}
	\label{conj:n*log n}
	Is it true that the optimal size for global certification of \textsc{Bipartiteness} is $\Theta(n \log n)$ ?
\end{restatable}

In other words, they conjectured that there is no way to compress the identifiers in the global certificate. Here, we prove that this conjecture is false: one can actually use only $\Theta(n)$ bits in this setting (in fact, we prove a more general result in terms of graph homomorphisms, see Theorem~\ref{thm:graph_homomorphism}).
To prove it, we used a technique of perfect hashing. Intuitively, for bipartiteness, the global certificate consists in (1) a bijection $h$ between the set of identifiers and $[n]$, that can be encoded compactly, and (2) a list of $n$ bits, where the $i$-th bit corresponds to the color in a proper 2-coloring of the vertex whose identifier is mapped to $i$ by the bijection~$h$.
The key ingredient is Theorem~\ref{thm:perfect_hashing}, which ensures that such a bijection~$h$ (called a \emph{perfect hash function}) can be encoded compactly.

This technique of perfect hashing does not consist precisely in a renaming as we will define it in Section~\ref{sec:definitions}, but the idea is similar, namely compress the identifiers in range~$[n]$. We also present another application of this hashing tool in Theorem~\ref{thm:distinctness}, which gives an efficient local certification scheme for the \textsc{Distinctness} problem.

%%%%%%%%%%%%%%%%%%%%%%%%
\section{Models and definitions}
\label{sec:definitions}
%%%%%%%%%%%%%%%%%%%%%%%%

For completeness, we now introduce some basic graph definitions. All the graphs we consider are finite, simple, non-oriented, and connected. Let $G$ be a graph. The vertex set and edge set are denoted by $V(G)$ and $E(G)$, respectively. For every $u \in V(G)$, we denote by $N(u)$ the \emph{open neighborhood} of $u$, which is set of vertices $v \in V$ such that $uv \in E(G)$, and by $N[u]$ the \emph{closed neighborhood} of $u$, which is $N(u) \cup \{u\}$.
A \emph{homomorphism} from a graph $G$ to a graph $H$ is a function $\varphi : V(G) \rightarrow V(H)$ such that, for every edge $uv \in E(G)$, we have $\varphi(u) \varphi(v) \in E(H)$. Note that homomorphisms generalize colorings, since a graph $G$ is $k$-colorable if and only if exists a homomorphism from $G$ to the complete graph on $k$ vertices (in particular, $G$ is bipartite if and only if there exists a homomorphism from $G$ to an edge).

A \emph{graph with inputs} is a pair $(G, i)$ where $G$ is a graph, and $i$ is a mapping $V(G) \rightarrow \mathcal{I}$, where $\mathcal{I}$ is a set called the set of \emph{inputs} and $i$ the \emph{input function}.
A \emph{renaming with range $\mathcal{I}$} of a graph $G$ is an \emph{injective} input function $i : V(G) \rightarrow \mathcal{I}$. In particular, for any $n$-vertex graph $G$, a renaming with range $[n]$ is precisely a permutation $V(G)\to [n]$.

Now, let us define formally the model of certification. Let $M : \NN \rightarrow \NN$, called the \emph{identifier range} (which is fixed: it is part of the framework for which certification schemes will be designed). Let $n = |V(G)|$. In the following, we just write $M$ instead of $M(n)$ for the sake of readability.
An \emph{identifier assignment} of $G$ is an injective mapping \mbox{$\mathrm{Id} : V \rightarrow [M]$}.
Finally, let $C$ be a set, called the set of \emph{certificates}.
We will define two different models of certification, namely \emph{local} and~\emph{global} certification.

\begin{definition}[Local certification]
    In local certification, a \emph{certificate assignment} for a graph $G$ (possibly equipped with an input function $i$) is a mapping $c : V(G) \rightarrow C$. Given an identifier assignment $\mathrm{Id}$, a certificate assignment~$c$, and a vertex $u \in V(G)$, the \emph{view of $u$} consists in all the information available in its neighborhood, that~is:
    \begin{itemize}
        \item its own identifier $\mathrm{Id}(u)$;
        \item its own input $i(u)$ (if $G$ is a graph with inputs); 
		\item its own certificate $c(u)$;
        \item the set of identifiers, inputs (if $G$ is a graph with inputs) and certificates of its neighbors, which is $\{(\mathrm{Id}(v), i(v), c(v)) \; | \; v \in N(u)\}$.
    \end{itemize}
\end{definition}

\begin{definition}[Global certification]
    In global certification, a \emph{global certificate} is a certificate $c \in C$. Given a graph $G$ (possibly equipped with an input function $i$), an identifier assignment $\mathrm{Id}$, a global certificate~$c$, and a vertex $u \in V(G)$, the \emph{view of $u$} consists in all the information available in its neighborhood, that~is:
    \begin{itemize}
        \item its own identifier $\mathrm{Id}(u)$;
        \item its own input $i(u)$;
		\item the global certificate~$c$;
        \item the set of identifiers and inputs (if $G$ is a graph with inputs) of its neighbors, which is $\{(\mathrm{Id}(v), i(v)) \; | \; v \in N(u)\}$.
    \end{itemize}
\end{definition}

A \emph{verification algorithm} is a function which takes as input the view of a vertex, and outputs a decision (\emph{accept} or \emph{reject}).

Let $\mathcal{P}$ be a property on graphs (possibly with inputs). We say that there is a local certification scheme (resp.\ a global certification scheme) with size $s(n)$ and identifier range~$M$ if there exists a verification algorithm $A$ such that, for all $n \in \NN$, there exists a set $C$ of size $2^{s(n)}$ (equivalently, the elements of $C$ can be seen as all binary words of $s(n)$ bits) satisfying the following condition: for every graph $G$ (possibly equipped with an input function~$i$) with $n$ vertices, $G$ satisfies $\mathcal{P}$ if and only if, for every identifier assignment $\mathrm{Id}$ with range $M$, there exists a certificate assignment $c : V(G) \rightarrow C$ (resp. a global certificate $c \in C$) such that $A$ accepts on every vertex.
We will often say that the certificates are given by a \emph{prover} which, intuitively, tries to convince the vertices that $\mathcal{P}$ is satisfied, but who can succeed only if $G$ indeed satisfies $\mathcal{P}$. 

Note that a verification algorithm is just a function, with no more requirements. In particular, it does not have to be decidable. However, in practice, when designing a certification scheme to prove upper bounds, it turns out to be decidable and often computable in polynomial time. The fact that no assumptions are made on this verification function in the definition just strengthens the results when proving lower bounds, by showing that it does not come from computational limits.

\medskip

We conclude these preliminaries with the definition of perfect hashing and a classical result on the size of perfect hash families.
\begin{definition}
	\label{def:perfect_hashing}
	Let $k, \ell \in \NN$ with $k \leqslant \ell$, and let $\mathcal{H}$ be a set of functions \mbox{$[\ell] \rightarrow [k]$}.
	\begin{enumerate}
		\item A function $h \in \mathcal{H}$ is a \emph{perfect hash function} for $S \subseteq [\ell]$ if $h(x) \neq h(y)$ for all $x, y \in S$, $x \neq y$.
		\item The family of functions $\mathcal{H}$ is a $(k, \ell)$-perfect hash family if, for every $S \subseteq [\ell]$ with $|S|=k$, there exists $h \in \mathcal{H}$ which is perfect for $S$.
	\end{enumerate}
\end{definition}

 We will need the following Theorem~\ref{thm:perfect_hashing} (see e.g. \cite{Mehlhorn84}, Chapter~3, Theorem~7 for a proof).

\begin{theorem}
	\label{thm:perfect_hashing}
	Let $k, \ell \in \NN$ with $k \leqslant \ell$. There exists a $(k, \ell)$-perfect hash family $\mathcal{H}_{k, \ell}$ which has size $\lceil k e^k \log \ell \rceil$.
\end{theorem}

A consequence of Theorem~\ref{thm:perfect_hashing} is that, to encode a hash function $h : [M] \rightarrow [n]$ which is perfect for a given set $S \subseteq M$ of size $n$ (for instance, the set of identifiers of a $n$-vertex graph), we just need $\log \lceil n e^n \log M\rceil = O(n + \log \log M)$ bits. See Theorems~\ref{thm:distinctness} and~\ref{thm:graph_homomorphism} for examples of applications.

%%%%%%%%%%%%%%%%%%%%%%%%%%%%%%%%%%%%%%%%%%%
\section{Renaming theorems}
\subsection{Certifying some renaming to \texorpdfstring{$[n]$}{[n]} with \texorpdfstring{$O(\log M)$}{O(log M)} bits}
\label{sec:Seb}
%%%%%%%%%%%%%%%%%%%%%%%%%%%%%%%%%%%%%%%%%%%

In this section we prove Theorem \ref{thm:seb}, which we restate here for convenience.

%%%%%%%%%%
\ThmSeb*
%%%%%%%%%

\begin{proof}
    More precisely, we will prove that the following holds: there exists a local certification scheme with $O(\log M)$ bits such that
    \begin{enumerate}
        \item if all the vertices accept with some certificate assignment, then the renaming written in the input is a correct renaming of range $[n]$, and
        \item there exists a renaming of range $[n]$ such that all the vertices accept with some certificate assignment.
    \end{enumerate}

    Recall that for an $n$-vertex graph $G$, a renaming with range $[n]$ is precisely a permutation $V(G)\to [n]$. So the two items above can be rewritten as: (1) if all vertices accept, then the input is a permutation of $ [n]$ and (2) there exists a permutation of $[n]$ such that if it is given in input, then all vertices accept.
    
    \subsection*{Description of a renaming that can be certified with $O(\log M)$ bits.}
    Let $G$ be an $n$-vertex graph. For every $u \in V(G)$, let us denote by $\mathrm{Id}(u)$ the original identifier of~$u$, and by $\renaming(u)$ the renaming of $u$ that is written in its input.
    Let us describe a renaming in range $[n]$ that can be certified using $O(\log M)$ bits.
    Let $\T$ be a spanning tree of~$G$, rooted at some vertex~$r$. For every vertex $u \in V(G)$, let us denote by $\T_u$ the subtree of~$\T$ rooted at $u$, and by $|\T_u|$ its number of vertices.
    We set $\renaming(r):=1$. Now, let $u \in V(G)$ with children  $v_1, \ldots, v_k$ in~$\T$ ordered by increasing (original) identifiers $\mathrm{Id}(v_1), \ldots, \mathrm{Id}(v_k)$. We define $\renaming(v_1):=\renaming(u)+1$, and for every $i \in \{2, \ldots, k\}$, $\renaming(v_i):=\renaming(v_{i-1})+|\T_{v_{i-1}}|$.
    In other words, this renaming corresponds to the order in which the vertices would be visited in a depth-first-search in~$\T$, starting at~$r$. See Figure~\ref{fig:1} for an example.

		\begin{figure}[h]
			\centering
			
			\begin{tikzpicture}
				\scriptsize
				\node[circle, draw] (1) at (-1.5, 0) {1};
				\node[circle, draw] (2) at (-6, -2) {2};
				\node[circle, draw] (9) at (-3, -2) {9};
				\node[circle, draw] (15) at (0, -2) {15};
				\node[circle, draw] (16) at (3, -2) {16};
				\node[circle, draw] (3) at (-7, -4) {3};
				\node[circle, draw] (6) at (-6, -4) {6};
				\node[circle, draw] (7) at (-5, -4) {7};
				\node[circle, draw] (10) at (-3.5, -4) {10};
				\node[circle, draw] (11) at (-2.5, -4) {11};
				\node[circle, draw] (17) at (2.5, -4) {17};		
				\node[circle, draw] (18) at (3.5, -4) {18};
				\node[circle, draw] (4) at (-7.5, -6) {4};
				\node[circle, draw] (5) at (-6.5, -6) {5};
				\node[circle, draw] (8) at (-5, -6) {8};
				\node[circle, draw] (12) at (-3.5, -6) {12};
				\node[circle, draw] (13) at (-2.5, -6) {13};
				\node[circle, draw] (14) at (-1.5, -6) {14};
				\node[circle, draw] (19) at (3, -6) {19};
				\node[circle, draw] (20) at (4, -6) {20};
				
				\draw (1) edge node{} (2);
				\draw (1) edge node{} (9);
				\draw (1) edge node{} (15);
				\draw (1) edge node{} (16);
				\draw (2) edge node{} (3);
				\draw (2) edge node{} (6);
				\draw (2) edge node{} (7);
				\draw (3) edge node{} (4);
				\draw (3) edge node{} (5);
				\draw (7) edge node{} (8);
				\draw (9) edge node{} (10);
				\draw (9) edge node{} (11);
				\draw (11) edge node{} (12);
				\draw (11) edge node{} (13);
				\draw (11) edge node{} (14);
				\draw (16) edge node{} (17);
				\draw (16) edge node{} (18);
				\draw (18) edge node{} (19);
				\draw (18) edge node{} (20);
				
				\node at (-1.9, 0.2) {\textcolor{red}{20}};
				\node at (-6.35, -1.8) {\textcolor{red}{7}};
				\node at (-3.35, -1.8) {\textcolor{red}{6}};
				\node at (-0.4, -1.8) {\textcolor{red}{1}};
				\node at (3.45, -1.8) {\textcolor{red}{5}};
				\node at (-7.35, -3.8) {\textcolor{red}{3}};
				\node at (-6.35, -3.8) {\textcolor{red}{1}};
				\node at (-5.35, -3.8) {\textcolor{red}{2}};
				\node at (-3.9, -3.8) {\textcolor{red}{1}};
				\node at (-2.9, -3.8) {\textcolor{red}{4}};
				\node at (2.1, -3.8) {\textcolor{red}{1}};
				\node at (3.1, -3.8) {\textcolor{red}{3}};
				\node at (-7.85, -5.8) {\textcolor{red}{1}};
				\node at (-6.85, -5.8) {\textcolor{red}{1}};
				\node at (-5.35, -5.8) {\textcolor{red}{1}};
				\node at (-3.9, -5.8) {\textcolor{red}{1}};
				\node at (-2.9, -5.8) {\textcolor{red}{1}};
				\node at (-1.9, -5.8) {\textcolor{red}{1}};
				\node at (2.6, -5.8) {\textcolor{red}{1}};
				\node at (3.6, -5.8) {\textcolor{red}{1}};
			\end{tikzpicture}
			
			\caption{Illustration of the renaming. The edges which are drawn here are those of $\T$. For each internal node, its children are written in increasing order of the original identifiers. The integer in every node $u$ corresponds to $\renaming(u)$. The red integer near a node $u$ corresponds to $|\T_u|$. (The original identifiers are omitted.)}
			\label{fig:1}
		\end{figure}
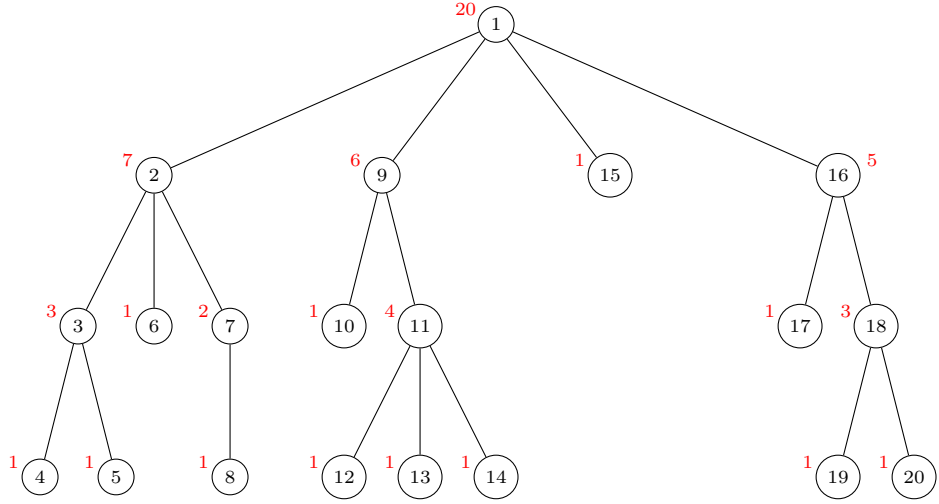

    \subsection*{Certification of the renaming.}
    Let us describe the certificates given by the prover to the vertices. Let $\mathcal{T}$ be the spanning tree of $G$ which gives the renaming described previously.
    %\textcolor{red}{On correct instances, it does not choose it, right? Or maybe there can be several good trees?}
    For every vertex $u \in V(G)$, we denote by $\T_u$ the subtree of $\T$ rooted at $u$. The certificate of every vertex~$u$ consists in four parts, denoted by $\idroot[u]$, $\idparent[u]$, $\distance[u]$, $\subtree[u]$, which are defined as follows.
    \begin{itemize}
        \item In $\idroot[u]$, the prover writes the identifier $\mathrm{Id}(r)$ of $r$.
        \item In $\idparent[u]$, the prover writes the identifier $\mathrm{Id}(w)$ of the parent $w$ of~$u$ in~$\T$ (if $u \neq r$) and $\mathrm{Id}(r)$ otherwise.
        \item In $\distance[u]$, the prover writes the distance from $u$ to $r$ in $\T$.
        \item In $\subtree[u]$, the prover writes the number of vertices in $\T_u$.
    \end{itemize}
    
    Note that all of these four parts can be encoded with $O(\log M)$ bits (because $n \leqslant M$). Thus, the overall size of the certificates is $O(\log M)$.
    
    \subsection*{Verification of the certificates.}
    Let us now explain how the vertices check the correctness of their certificates. If $u$ and $v$ are two neighbors such that $\idparent[v]=\mathrm{Id}(u)$, we say that $u$ is the \emph{parent} of $v$, or equivalently that $v$ is a \emph{child} of $u$. If $u$ does not have any children, we say that $u$ is a \emph{leaf}. The verification procedure of every vertex $u \in V(G)$ consists in several steps, which are the following ones. At each step, if the verification fails, $u$ rejects. If $u$ did not reject at any step, it finally accepts.
    \begin{enumerate}[(i)]
        \item First, $u$ checks that all its neighbors have the same root, that is $\idroot[u]=\idroot[v]$ for every neighbor $v$ of $u$.

        If no vertex rejects at this point, then $\idroot$ is the same in the certificates of all the vertices, so in the following we will denote by $\idroot$ this common value written in the certificates.

        \item If $\mathrm{Id}(u) \neq \idroot$, $u$ checks that it is a neighbor of its parent~$w$, and it also checks that $\distance[w]=\distance[u]-1$. 

        \item If $u$ is a leaf, it checks that $\subtree[u]=1$.
        Otherwise, $u$ checks that $\subtree[u] = 1 + \Sigma_{i=1}^{k} \subtree[v_i]$, where $v_1, \ldots, v_k$ are the children of $u$. 

        \item Finally, $u$ checks the correctness of the renaming. Namely, if $\mathrm{Id}(u) = \idroot$, $u$ checks that $\renaming[u]=1$. If $u$ is not a leaf and has $v_1, \ldots, v_k$ as children, where the identifiers $\mathrm{Id}(v_1), \ldots, \mathrm{Id}(v_k)$ are in increasing order, it checks that $\renaming[v_1]=\renaming[u]+1$, and that for all $i \in \{2, \ldots, k\}$, $\renaming[v_i]=\renaming[v_{i-1}]+\subtree[v_{i-1}]$.
    \end{enumerate}

\subsection*{Correctness.}

Let us prove the correctness of this certification scheme.
First, assume that the renaming is the one described above, and that the prover gives the correct certificates to the vertices.
Then, by definition of the renaming and the certificates, no vertex will reject in the verification procedure, so all the vertices accept.

Conversely, assume that all the vertices accept with some renaming and some certificates, and let us check that the renaming is correct.
First, note that $G$ contains a vertex having the identifier $\idroot$. Let $r$ be the vertex of $G$ such that $\distance[r]$ is minimal in~$G$. Since $r$ accepts, we have $\mathrm{Id}(r)=\idroot$, else $r$ would reject at step~(ii). Then, we can reconstruct the spanning tree~$\T$ that the prover used to assign the certificates: the root is~$r$, and a vertex $v$ is a child of $u$ if $u$ and $v$ are neighbors and $\idparent[v]=\mathrm{Id}(u)$. Step~(ii) ensures that this definition leads to a correct spanning tree~$\T$. Then, step~(iii) ensures that for each vertex~$u$, we have $\subtree[u]=|\T_u|$. Finally, step~(iv) guarantees that $\renaming$ corresponds to the order in which the vertices would be visited in a depth-first search starting at~$r$, which is indeed a correct renaming in $[n]$.
\end{proof}

%%%%%%%%%%%%%%%%%%%%%%%%%%%%%%%%%%%%%%%%%%%%
\subsection{Certifying any renaming to \texorpdfstring{$[n]$}{[n]} with \texorpdfstring{$O(n + \log M)$}{O(n+log M)} bits}
\label{sec:Louis}
%%%%%%%%%%%%%%%%%%%%%%%%%%%%%%%%%%%%%%%%%%%
In this section we prove Theorem \ref{thm:ThmLouis}, which we restate here for convenience.
%%%%%%%%%%
\ThmLouis*
%%%%%%%%%

\begin{proof}
More precisely, we prove the following: there exists a local certification scheme with $O(n + \log M)$ bits such that
    \begin{enumerate}
        \item if all the vertices accept with some certificate assignment, then the renaming written in the input is a correct renaming of range $[n]$, and
        \item for all renamings of range $[n]$, all the vertices accept with some certificate assignment.
    \end{enumerate}

    As before, the two items above can be rewritten in the language of permutation as follows: (1) if all vertices accept, then the input is a permutation of $ [n]$ and (2) for any permutation of $[n]$ in input, all vertices accept.

    \medskip

We will indeed prove the more general Theorem~\ref{thm:partition}, from which Theorem~\ref{thm:ThmLouis} follows as a simple consequence.

\begin{theorem}
    \label{thm:partition}
    Let $G$ be $n$-vertex graph. Assume that each vertex~$u$ of~$G$ is given as input a set of integers~$S_u \subseteq [n]$. Then, certifying that $\{S_u\}_{u \in V(G)}$ is a partition of $[n]$ (i.e., that these sets are pairwise disjoint and their union is equal to $[n]$) can be done with $O(n + \log M)$ bits.
\end{theorem}

Assuming Theorem~\ref{thm:partition}, the proof of Theorem~\ref{thm:ThmLouis} is immediate, since it is just the particular case where the input set~$S_u$ of each vertex~$u$ is a singleton. 

    \medskip

\noindent{\emph{Proof of Theorem \ref{thm:partition}}.} In the following, we assume that all the vertices know the number of vertices~$n$, since it can be certified with $O(\log n)$ bits (as in the proof of Theorem~\ref{thm:seb} for instance), and we aim for a certification scheme using $O(n + \log M)$ bits in total.

\subsection*{Certification.}

Let us describe the certificates given by the prover to the vertices. The prover chooses a spanning tree $\T$ of $G$ rooted at some vertex $r$. As in the proof of Theorem~\ref{thm:seb}, for each vertex~$u$ of~$G$, we denote by~$\T_u$ the subtree of $\T$ rooted at~$u$. The certificate of every vertex~$u$ consists in four parts $\idroot[u]$, $\idparent[u]$, $\distance[u]$, $\subtree[u]$. The parts $\idroot[u]$, $\idparent[u]$ and $\distance[u]$ are defined exactly as in the proof of Theorem~\ref{thm:seb}, and only $\subtree[u]$ differs. In $\subtree[u]$, the prover writes a binary vector of $n$ bits whose $j$-th bit is equal to~1 if and only if there exists a vertex $v \in \T_u$ such that $j \in S_v$. %\textcolor{red}{This is assuming that the certification is correct}

The three first parts of the certificates have size $O(\log M)$, and the $\subtree$ part has size $O(n)$, so the overall size of the certificates is $O(n + \log M)$.

\subsection*{Verification of the certificates.}

 Let us now explain how the vertices check the correctness of their certificates. The definitions of~\emph{parent}, \emph{child} and~\emph{leaf} are the same as in the proof of Theorem~\ref{thm:seb}. As in the proof of Theorem~\ref{thm:seb}, the verification procedure consists of several steps. Steps~(i) and~(ii) are identical, and there are finally three last steps~(iii), (iv) and~(v) which are the following:
 \begin{enumerate}
   
    \item[(iii)] $u$ checks that its $\subtree$ vector have size~$n$ (recall that we assumed that each vertex knows~$n$, since it can be certified with $O(\log n)$ additional bits)
     \item[(iv)] If $u$ is a leaf, it checks that $\subtree[u]_j=1$ if and only if $j \in S_u$ for every integer~$j$. Otherwise, let us denote by $v_1, \ldots, v_k$ the children of $u$. For every integer~$j$, the vertex $u$ checks that one of the three following cases holds:
     \begin{itemize}
         \item $\subtree[u]_j=0$, $j \notin S_u$, and for all $i \in [k]$, $\subtree[v_i]_j=0$, %\textcolor{red}{And it's not its new name}
         or
         \item $\subtree[u]_j=1$, $j \notin S_u$, and there exists a unique $i \in [k]$ such that $\subtree[v_i]_j=1$, or
         \item $\subtree[u]_j=1$, $j \in S_u$, and for all $i \in [k]$, $\subtree[v_i]_j=0$.
     \end{itemize}
     \item[(v)] If $\mathrm{Id}(u)=\idroot$, $u$ checks that $\subtree[u]_j=1$ for all index~$j$.
 \end{enumerate}

\subsection*{Correctness.}

Let us prove the correctness of this certification scheme. First, as in the proof of Theorem~\ref{thm:seb}, if $[n]$ is indeed the disjoint union of the sets $S_u$ given in input, and if the prover gives the certificates as described above, no vertex will reject by definition.

Conversely, assume that all the vertices accept with some sets $S_u$ as input and some certificates. Let us check that $[n]$ is indeed equal to the disjoint union of the sets $S_u$. Exactly as in the proof of Theorem~\ref{thm:seb}, step~(ii) ensures that there exists a vertex $r$ such that $\mathrm{Id}(r)=\idroot$, and we can reconstruct a spanning tree~$\T$ rooted at $r$ that the prover used to assign the certificates. Then, step~(iii) ensures that $\subtree$ is indeed a $n$-bit vector in the certificate of each vertex. Then, step~(iv) ensures that $\subtree[u]_j$ is equal to~1 if and only if there exists a unique $v \in \T_u$ such that $j \in S_v$. Since all the bits of $\subtree[r]$ are equal to~1, it implies that for every $j \in \{1, \ldots, n\}$, there exists a unique vertex $u$ such that $j \in S_u$. Thus, the union of the sets $S_u$ is disjoint, and equal to~$[n]$.
\medskip

This concludes the proof of Theorem \ref{thm:partition}.
\end{proof}

%%%%%%%%%%%%%%%%%%%%%%%%%%%%%%%%%%%%%%%%%%%

%%%%%%%%%%%%%%%%%%%%%%%%%%%%%%%%%%%%%%%%%%%

\section{Applications}
\label{sec:applications}

\subsection{Obtaining optimality in the \texorpdfstring{$\Theta(n)$}{Theta(n)}-bit regime}\label{subsec:optimal-theta-n}

In this subsection, we focus on local certification and prove that our results can be used to 
obtain optimal certification schemes for several classical problems.

\subsubsection*{\texorpdfstring{$K_\ell$}{K_l}-freeness}
Here we make the classical assumption that, unless specified otherwise, the range of identifiers is polynomial in~$n$, that is $M = O(n^c)$ for some $c > 0$. It was proved\footnote{The model considered in \cite{DKO13} is a bit different from local certification, but their proof can easily be adapted to work in this model.} in \cite[Theorem 15]{DKO13} that for any fixed $\ell\ge 4$, locally certifying that an $n$-vertex graph $G$ is $K_\ell$-free (meaning that $G$ does not contain $K_\ell$ as a subgraph) requires certificates of size $\Omega(n)$. On the other hand, $K_\ell$-freeness can easily be locally certified with certificates of $O(n\log n)$ bits: it suffices to assign to each vertex  the list of (identifiers of) its neighbors as a certificate, and then the vertices can readily check the validity of these certificates, and the local view of each vertex $v$ is enough to compute the size of a maximum clique containing $v$ in $G$ (and thus certify that $G$ is $K_\ell$-free, if it is the case).  Storing the list of identifiers of its neighbors can cost up to $\Omega(n\log n)$ bits in the worst case, so with this technique it does not seem to be possible to match the lower bound of $\Omega(n)$ from \cite{DKO13}. We now explain how to obtain an optimal local certification scheme for $K_\ell$-freeness ($\ell$ fixed) with certificates of size $O(n)$, as a simple consequence of Theorem \ref{thm:seb} or \ref{thm:ThmLouis}.

\begin{theorem}
\label{thm:cliques}
For any fixed integer $\ell$, $K_\ell$-freeness can be locally certified with certificates of $O(n)$ bits in  $n$-vertex graphs (and this is optimal).
\end{theorem}

\begin{proof}
Start by renaming the identifiers of $G$ in the range $[n]$ using Theorem \ref{thm:seb} or \ref{thm:ThmLouis}. In either case, the renaming can be locally certified with $O(n)$ bits. 
Now, the certificate of $u$ consists of a binary vector of size $[n]$ whose $j$th coordinate is $1$ if and only if $u$ is adjacent to the vertex whose new identifier is $j$. That is, we encode the neighborhood of $u$ with a vector of size $n$.
As explained above, this information is enough to  decide whether $G$ is $K_\ell$-free, as desired.
\smallskip

    The fact that the size of the certificates is optimal was proved in \cite{DKO13}.
\end{proof}

\subsubsection*{Diameter \texorpdfstring{$\leq \ell$}{at most l}}
Here we make again the assumption that the range of identifiers is polynomial in $n$, that is $M = O(n^c)$ for some $c > 0$. It was proved in \cite{Censor-HillelPP20} that, for any constant~$\ell$, having diameter at most $\ell$ can be certified locally with certificates of $O(n \log n)$ bits. Moreover, they obtained a lower bound of $\Omega(n)$ bits on the certificates for this problem. We obtain the following improvement.

\begin{theorem}
For any fixed integer $\ell$, having diameter at most $\ell$ can be locally certified with certificates of $O(n)$ bits in  $n$-vertex graphs (and this is optimal).
\end{theorem}

\begin{proof}
The upper bound technique in \cite{Censor-HillelPP20} is the following. In a correct instance, every node~$v$ is given as a certificate a table of size $n$. Each entry corresponds to a  node $w$ of the graph and contains the identifier of $w$ and the distance from $v$ to $w$. Checking these tables consists in checking in parallel $n$ BFS trees. More formally, for every vertex $w$, $v$ checks that $d(v,w) \leqslant \ell$, and that least one of its neighbors has distance one less to $w$. For every neighbor~$u$ of~$v$, $v$ also checks that $|d(u,w)-d(v,w)| \leqslant 1$.  Every entry can be encoded on $O(\log n+\log \ell)=O(\log n)$ bits, hence the result. 

In our scheme, we start by renaming the identifiers of $G$ in the range $[n]$ using Theorem \ref{thm:seb} or \ref{thm:ThmLouis}. In either case, the renaming can be locally certified with $O(n)$ bits.
Then on a correctly renamed instance, the prover assigns to every node $v$ a table $T_v$ of $n$ entries, where the entry $T_v[i]$ contains the distance from $v$ to the node with new identifier $i$. The verification can be done as in \cite{Censor-HillelPP20}, and the size of the certificates is $O(n)$.
\end{proof}

\subsubsection*{Distinctness}
In the introduction we have mentioned the \textsc{Distinctness} problem: given a graph~$G$ with identifiers of range $[M]$, and an input function~$i$ of range $[m]$, the goal is to certify that the inputs of all the vertices are distinct (i.e.\ that $i$ is injective). The integer $m$ is not known by the vertices in advance, and it is not a function of~$n$ (recall also that $M$ is a function of~$n$). See \cite{NPY20} for a study of this problem in the interactive version of  local certification. We prove the following result:

\begin{theorem}
\label{thm:distinctness}
The \textsc{Distinctness} property can be locally certified in $n$-vertex graphs with identifier range~$[M]$ and input range~$[m]$ with certificates of size~$O(n + \log M + \log \log m)$.
\end{theorem}

Note that Theorem~\ref{thm:distinctness} is in fact a generalization of Theorem~\ref{thm:ThmLouis}. Indeed, Theorem~\ref{thm:ThmLouis} is just the particular case of Theorem~\ref{thm:distinctness} where $m=n$.

\begin{proof}
    The certificates given by the prover to the vertices consist in three parts.
    First, the prover writes $\lceil \log m \rceil$ in the certificate of every vertex, using $O(\log \log m)$ bits.
    Then, let $m' := 2^{\lceil \log m\rceil}$. We have $m \leqslant m' \leqslant 2m$. The prover chooses a hash function $h : [m'] \rightarrow [n]$ which is perfect for the set of inputs $\{i(u) \; | \; u \in V(G)\}$, and writes this hash function in the certificate of every vertex. By Theorem~\ref{thm:perfect_hashing}, it uses $O(n + \log \log m)$ bits (because $m' \leqslant 2m)$. Finally, the prover uses the certification of Theorem~\ref{thm:ThmLouis} to certify that $h \circ i$ is a correct renaming, which uses $O(n + \log M)$ bits. In total, the certificates are of size $O(n + \log M + \log \log m)$.
    
    The verification procedure of each vertex $v$ is the following.
    First, it checks that the value of $\lceil \log m \rceil$ written in its certificate is the same as in the certificates of its neighbors. Then, it computes $m'$ and checks that $i(v)\le m'$.
    Finally, the last part of the certification just consists in checking whether it received the same hash function~$h$ as its neighbors, and it then performs the verification of Theorem~\ref{thm:ThmLouis} to check that~$h \circ i$ is indeed injective.
    
    Note that the vertices do not need to verify that the value of $\lceil \log m\rceil$ written in its certificate is correct: indeed, the only thing which matters is that they all compute the same integer $m'$ (and that $m'$ is an upper bound for the set of inputs). Note also that the vertices do not need to check that $h$ is perfect for the set of inputs: indeed, if all the vertices accept, even if $h$ is not perfect, all the inputs are distinct since $h \circ i$ is injective.
\end{proof}

\subsubsection*{Anti-triangle-freeness}

Here we make again the assumption that the range of identifiers is polynomial in $n$, that is $M = O(n^c)$ for some $c > 0$.

\begin{theorem}
    \label{thm:anti_triangle}
    Certifying locally that a graph does not have an independent set of size~$3$ can be done with certificates of size~$O(n)$.
\end{theorem}

Note that certifying anti-triangle-freeness (i.e.\ not containing an independent set of size $3$) seems at first sight to be more difficult than certifying triangle-freeness. Indeed, for the latter it is sufficient to write the adjacency of each vertex in its certificate, which can be done with $O(n)$ bits using renaming (see the proof of Theorem~\ref{thm:cliques}), but this is not sufficient anymore for anti-triangle-freeness. However, it can still be done with $O(n)$ bits using renaming techniques.

\begin{proof}
    First, note that a graph~$G$ does not have an independent set of size~$3$ if and only if, for every vertex~$u$, the non-neighborhood of~$u$ (that is, the set $V(G) \setminus N[u]$) is a clique (possibly empty in the case where $u$ is a universal vertex).
    
    Assume that~$G$ does not have an independent set of size~$3$, and let us describe the certificates given by the prover to the vertices. First, the prover starts by certifying a renaming in the range $[n]$ using Theorem~\ref{thm:seb} or~\ref{thm:ThmLouis}, using $O(n)$ bits. Then, it writes in the certificate of every vertex its adjacency vector, that is, a binary vector of size~$n$ whose $j$th coordinate is equal to~$1$ if and only if~$u$ is adjacent to the vertex whose new identifier is~$j$. This also takes $O(n)$ bits.
    Finally, the prover does the following:
    \begin{enumerate}
        \item If there exists a universal vertex in~$G$, it chooses one universal vertex and writes its identifier in the certificates of all the vertices.

        \item If there is no universal vertex in~$G$, then for every $u \in V(G)$, $V(G) \setminus N[u]$ is non-empty, and the prover chooses a vertex~$\pi(u) \in V(G) \setminus N[u]$. For every $v \in V(G)$, let \mbox{$S_v:=\{u \in V(G) \; | \; v = \pi(u)\}$}. The prover then writes, in the certificate of each vertex~$v$, a binary vector denoted by $P[v]$ such that, for every $i \in [n]$, $P[v]_i=1$ if and only if $u \in S_v$, where~$u$ is the vertex whose new identifier is~$i$. Finally, the prover certifies that the sets $(S_v)_{v \in V(G)}$ form a partition of $[n]$ using Theorem~\ref{thm:partition}. In total, this takes~$O(n)$ bits.
    \end{enumerate}

    The verification of the vertices consists in doing the following. Every vertex checks that the renaming is correct, and that its adjacency vector is correctly written in its certificate. If no vertex rejects at this point, then the renaming is correct and every vertex knows the neighborhood of each of its neighbors. Now, there are two cases:

    \begin{enumerate}
        \item If the prover wrote in the certificates that there is a universal vertex~$u$ together with its identifier, every vertex checks that the identifier of the universal vertex $u$ written in its certificate is the same as in the certificate of its neighbors, and that it indeed sees~$u$. If no vertex rejects, then~$u$ is indeed a universal vertex, so $u$ knows all the edges of the graph (because every vertex has its adjacency vector written in its certificate) and rejects if there is an independent set of size~3.

        \item If the prover wrote that there is no universal vertex in the graph, then the vertices first check that $[n]$ is indeed the disjoint union of $\{S_v\}_{v \in V(G)}$ using Theorem~\ref{thm:partition}. Then, remember that our aim is to check that for every vertex~$u$, the non-neighborhood of $u$ is a clique. To do so, for every pair of vertices $u$ and~$v$, the vertex $v$ checks that
        it is a non-neighbor of $u$ if and only if $v$ is a neighbor of $\pi(u)$ or $v$ is $\pi(u)$ itself (note that, for every neighbor~$w$ of~$v$, $v$ can determine whether $w = \pi(u)$ because this is equivalent to $u \in S_w$, and the set $S_w$ is written in the certificate of~$w$).
        Finally, in the case where $v = \pi(u)$, if no vertex rejected at this point, $v$ is adjacent to all the non-neighbors of~$u$, and can thus determine the subgraph induced by $V(G) \setminus N[u]$ (recall that every vertex knows the neighborhood of its neighbors), and reject if it is not a clique.
    \end{enumerate}
    Finally, every vertex which did not reject previously accepts. 

    \smallskip
    
    The certification scheme is correct. Indeed, if all the vertices accept, either there is a universal vertex~$u$ and in this case there is no independent set of size~3 (because $u$ would have rejected), or there is no universal vertex, and in this case, for every vertex~$u \in V(G)$, the vertex $v = \pi(u)$ sees all the non-neighbors of $u$ (otherwise some non-neighbor of~$u$ would have rejected), and~$v$ checked that $V(G) \setminus N[u]$ induces a clique.
\end{proof}

Note that the size of the certificates in Theorem \ref{thm:anti_triangle} is close to  optimal. By adapting a result of \cite{DKO13} proved in the congested clique model, we show in Appendix \ref{sec:app} that triangle-freeness and anti-triangle-freeness both require certificates of $n/e^{O(\sqrt{\log n})}$ bits. 

%Indeed, it was proved in \cite[Corollary 25]{DKO13} that $n/e^{O(\sqrt{\log n})}$ bits are necessary to certify that a graph does not have a triangle. As the model considered in \cite{DKO13} is a bit different from local certification, we provide  a full proof of the result in this setting in Appendix \ref{sec:app}, for future reference. Their construction is based on the existence of an $n$-vertex graph that has $n^2/e^{O(\sqrt{\log n})}$ triangles, and has the property that each edge belongs to exactly one triangle. By considering the complement of this graph, we obtain a connected $n$-vertex graph that has $n^2/e^{O(\sqrt{\log n})}$ anti-triangles, and such that each non-edge belongs to exactly one anti-triangle. The rest of their proof applies verbatim and shows a lower bound of $n/e^{O(\sqrt{\log n})}$ on the size of the certificates for anti-triangle-freeness. 

\subsection{Local certification using universal graphs}

In this subsection, we show how the existence of a universal graph for a monotone graph class can lead to a sub-quadratic certification scheme. Our argument crucially uses renaming and hashing (in fact, it uses the certification procedure for the \textsc{Distinctness} problem presented in Theorem~\ref{thm:distinctness}, which itself uses renaming and hashing). We make again the assumption that the identifier range is polynomial in~$n$, that is $M = O(n^c)$ for some $c > 0$.

A class of graphs~$\G$ is \emph{monotone} if it is closed under taking (non-necessarily induced) subgraphs, that is: for all $G \in \G$, if $G'$ is a subgraph of $G$ (i.e.\ a graph obtained from~$G$ by deleting some set of vertices and edges), then $G' \in \G$. Example of monotone classes include planar graphs, and more generally minor-closed classes, as well as graphs avoiding some fixed subgraph (for instance graphs with no 4-cycle).

\begin{theorem}
    \label{thm:universal}
    Let $\G$ be a monotone class of graphs and let $s : \NN \to \NN$ be a function with the following property: for every $n \in \NN$, there exists an $s(n)$-vertex graph $U^\G_n \in \G$ such that every $n$-vertex graph $G \in \G$ is a subgraph of $U^\G_n$.

    Then, there exists a certification scheme for the property of belonging to $\G$ with certificates of size~$O(n + \log s(n))$.
\end{theorem}

\begin{proof}
    First, note that since $\G$ is monotone and $U^\G_n \in \G$, for every $n \in \NN$, an $n$-vertex graph~$G$ belongs to the class~$\G$ if and only if $G$~is a subgraph of $U^\G_n$.
    
    For every $n \in \NN$, fix a bijection~$i_n : V(U^\G_n) \to [s(n)]$ and consider it as a unique identifier assignment to the vertices of $U^\G_n$.
    This bijection~$i_n$ is fixed before the certification procedure.
    
    We now describe the certificates assigned to the vertices of an $n$-vertex graph $G \in \G$ by the prover. The prover chooses a subgraph $G'$ of~$U^\mathcal{G}_n$ that is an isomorphic copy of~$G$, and identifies the function $\varphi : V(G) \to V(U^\mathcal{G}_n)$ that maps each vertex of $G$ to the corresponding vertex of $G'$ (viewed as a subgraph of $U^\mathcal{G}_n$). For every vertex $v \in V(G)$, the prover writes the identifier $i_n(\varphi(v))$ in the certificate of~$v$, using $O(\log s(n))$ bits.
    Then, using the certification scheme of Theorem~\ref{thm:distinctness}, the prover certifies that the values $i_n(\varphi(v))$ that it wrote in the certificates are distinct for all the vertices $v \in V(G)$, using $O(n + \log \log s(n))$ additional bits. In total, the certificates have size~$O(n + \log s(n))$.

    Let us describe how the vertices check their certificates. First, each vertex
    applies the verification procedure of Theorem~\ref{thm:distinctness}.
    If no vertex rejects at this point, then the identifiers $i_n(\varphi(v))$ are distinct for all $v \in V(G)$, so $\varphi$ is an injective function from $V(G)$ to $V(U^\G_n)$. Then, each vertex $v \in V(G)$ checks that, for each of its neighbors $w$, $\{\varphi(v),\varphi(w)\}$ is an edge of $U^\G_n$, and $v$ accepts if and only if it is the case. 
    
    Finally, note that this certification scheme is correct, because if every vertex accepts, $G$ is a subgraph of $U^\G_n$, so $G \in \G$. Conversely, if $G \in \G$, no vertex will reject with the certificates given as described previously.
\end{proof}

Note that Theorem~\ref{thm:universal} can be applied to obtain a non-trivial upper bound only for monotone classes~$\G$ which have a universal graph $U^\G_n \in \G$ with $2^{o(n^2)}$ vertices. Indeed, if the number of vertices of $U^\G_n \in \G$ is $2^{\Theta(n^2)}$, then the bound given by Theorem~\ref{thm:universal} is $O(n^2)$. This would not be useful, since a universal upper bound for any property is $O(n^2 + \log M)$ and in the present case where $M = O(n^c)$, this universal upper bound is just $O(n^2)$ (the idea of this universal upper bound is simply to certify a renaming to $[n]$ using Theorem~\ref{thm:seb} or~\ref{thm:ThmLouis}, and encoding the adjacency matrix of the graph in the certificate of every vertex).

Unfortunately, to our knowledge, there is no example of  a non-trivial monotone class~$\G$ in which it is known that such a universal graph with $2^{o(n^2)}$ vertices exists. Indeed, the research of universal graphs has received a considerable attention recently, but the constraint that the universal graph belongs to the class it represents is usually not satisfied. For example, \cite{EsperetJM23} builds a sparse universal graph containing all planar graphs of a given size, but this graph is not planar. 
This motivates the following question:

\begin{question}
    \label{question}
    What are the monotone graph classes~$\G$ having a universal graph $U^\G_n \in \G$ with $2^{o(n^2)}$ vertices?
\end{question}

An example of a class for which Theorem~\ref{thm:universal} could be useful is the class of $H$-free graphs (a graph $G$ is $H$-free if it does not contain the graph~$H$ as a subgraph, not necessarily induced). Note that such a class is monotone. For these classes, it is only known that if $H$ is a tree then a logarithmic certification exists (see \cite{BousquetCFPZ24+}, based on~\cite{KorhonenR17}), and that for cliques the optimal size is $\Theta(n)$ (as discussed in Subsection~\ref{subsec:optimal-theta-n}). More bounds are known for forbidden induced subgraphs, in particular for paths~\cite{BousquetCFPZ24+, M24}, but the associated classes are not monotone in general.

One way to obtain universal graphs for $H$-free classes could be the following. In general graphs, it is known that if we take a $2^n$-vertex graph uniformly at random, it will contain almost surely all the $n$-vertex graphs as subgraphs (in particular, if $\G$ is the class of all graphs, for every $n \in \N$ there exists a universal graph $U^\mathcal{G}_n$ for $n$-vertex graphs in $\mathcal{G}$ which has $2^n$ vertices)~\cite{BT81}. 
It would be interesting to determine if the analogous result holds for random $H$-free graphs. Theorem~\ref{thm:universal} would provide a direct application: if such a universal graph for $H$-freeness exists and has $2^{n^c}$ vertices for some $c < 2$, then there is a certification scheme for $H$-freeness with certificates of size~$O(n^c)$.

\subsection{Global certification of \texorpdfstring{$H$}{H}-homomorphism}
\label{sec:global}

In this subsection, we focus on global certification. The property we want to certify is $H$-\textsc{Homomorphism} (the existence of a homomorphism to a given graph $H$).
A particular case which has already been studied in~\cite{FeuilloleyH18} is \textsc{Bipartiteness} (it corresponds to the case where $H$ is an edge).
Note that there exists a local certification scheme for \textsc{Bipartiteness} using only one bit per vertex (where the certificate is the color in a proper 2-coloring, and the verification of every node just consists in checking that it received a different color from all its neighbors).
However, with a global certificate, it is less clear how to certify it optimally. 
In~\cite{FeuilloleyH18}, the authors proved the following upper and lower bounds:

\begin{theorem}
	\label{thm:existing_bounds}
	If $s$ denotes the optimal certificate size for the global certification of \textsc{bipartiteness} for $n$-vertex graphs whose identifiers are in the range $[M]$, then:
	$$s = \Omega(n + \log \log M) \qquad \text{and} \qquad s = O(\min\{M, n \log M\})$$
\end{theorem}

The authors of~\cite{FeuilloleyH18} also made the conjecture that their lower bound can be improved to match their upper bound. Namely, in the standard case where $M=O(n^c)$, they conjectured that the answer to following question is positive:

\ConjectureGlobalBipartiteness*

Here, we disprove this conjecture, by proving the following stronger result. The key ingredient used in the proof is perfect hashing (Theorem~\ref{thm:perfect_hashing}).

\begin{theorem}
    \label{thm:graph_homomorphism}
	For any graph $H$, there exists a global certification scheme for $H$-\textsc{Homo\-morphism} for $n$-vertex graphs whose identifiers are in the range $[M]$,  with a certificate of size $O(n \log |V(H)| + \log \log M)$.
\end{theorem}

\begin{corollary}
	\label{cor:new_bound}
	There exists a global certification scheme for \textsc{Bipartiteness} for $n$-vertex graphs whose identifiers are in the range $[M]$, with a certificate of size $O(n + \log \log M)$.
\end{corollary}

Note that, in the standard case where $M$ is polynomial in $n$, Corollary~\ref{cor:new_bound} gives a certificate of size~$\Theta(n)$. Note also that this bound remains~$\Theta(n)$ even in the case where $M=2^{2^{O(n)}}$.

\begin{proof}[Proof of Theorem~\ref{thm:graph_homomorphism}.]
	Let us describe a global certification scheme for the existence of a homomorphism to $H = (V',E')$ using a certificate of size $O(n \log n' + \log \log M)$, where $n'=|V'|$.
	First, since $H$ has $n'$ vertices, we can number these vertices from $1$ to $n'$ and write the index of a vertex of $H$ on $\log n'$ bits. For every $k, \ell \in \NN$ with $k \leqslant \ell$, by applying Theorem~\ref{thm:perfect_hashing}, we can number the functions in $\mathcal{H}_{k, \ell}$ from $1$ to $|\mathcal{H}_{k,\ell}|$. Thus, a function of $\mathcal{H}_{k, \ell}$ can be represented using $\log |\mathcal{H}_{k,\ell}| = O(k + \log \log \ell)$ bits.
	
    Let $G=(V,E)$ be an $n$-vertex graph, for which there exists a homomorphism $\varphi$ from $G$ to $H$, and let $\mathrm{Id}$ be an identifier assignment of $G$. The certificate given by the prover is the following one.
	Let us denote by $S:=\{\mathrm{Id}(v) \; | \; v \in V\}$ the set of identifiers appearing in~$G$. The set $S$ is included in $[M]$ and has size $n$.
	Let $h \in \mathcal{H}_{n, M}$ be a perfect hash function for $S$. By definition, the function $h$ induces a bijection between $S$ and $[n]$. Let $L$ be the list of size $n$ such that the $i$-th element of $L$, denoted by $L[i]$, is equal to $\varphi(v)$, where $v$ is the unique vertex in $V$ such that $h(\mathrm{Id}(v))=i$.
	The certificate given by the prover to the vertices is the triple $(n, h, L)$, where $h$ is represented by its numbering in $\mathcal{H}_{n, M}$. Since it uses $O(n \log n')$ bits to represent $L$ and $O(n + \log \log M)$ bits to represent~$h$, the overall size of the certificate is $O(n \log n' + \log \log M)$. 
	
	Let us describe the verification algorithm. Each vertex $u$ does the following. First, it reads $n$ in the global certificate and computes $M$ (recall that the identifier range $M$ is a function of~$n$). Then, it determines the hash function $h \in \mathcal{H}_{n, M}$ thanks to its numbering in the certificate. Finally, $u$ accepts if and only if, for all $v \in N(u)$, $\{L[h(\mathrm{Id}(u))],L[h(\mathrm{Id}(v))]\} \in E'$. If it is not the case, $u$ rejects.
	
	Let us prove the correctness of our scheme. First, assume that $G$ indeed admits  a homomorphism to~$H$. Then, by giving the certificate as described above, since $\varphi$ is a homomorphism, each vertex $u \in V$ accepts.
    Conversely, assume that every vertex accepts with some certificate~$c$, and let us prove that there exists a homomorphism from $G$ to~$H$. Since all the vertices accept, every vertex $u$ checked if $\{L[h(\mathrm{Id}(u))],L[h(\mathrm{Id}(v))]\} \in E'$ for every $v \in N(u)$, for some function $h$ which is written in $c$. Note that nothing ensures that $h$ is indeed a perfect hash function for the set $S$ of identifiers, but in fact, it is not necessary to check that $h$ is injective on $S$. Indeed, since every vertex $u$ accepted, then for every $v \in N(u)$, we have $\{L[h(\mathrm{Id}(u))], L[h(\mathrm{Id}(v))]\} \in E'$. So $\varphi(u):= L[h(\mathrm{Id}(u))]$ defines a homomorphism from~$G$ to~$H$. Thus, it proves the correctness of the scheme.
\end{proof}

We note that the idea of using perfect hashing has independently been used in~\cite{EHZ24} with another type of labeling, but to our knowledge, it is the first time that perfect hashing is used in distributed computing. We hope that this technique could have other applications in future works, in particular for problems related to space complexity.

\subsection*{Acknowledgments}
We thank William Kuszmaul for fruitful discussion on hashing. We also thank Ami Paz, Anup Rao and Alexander Sherstov for the discussions on the nondeterministic communication complexity of set disjointness in the number-on-forehead model. Finally, we thank the reviewers for their helpful comments and suggestions.

\bibliographystyle{plain}

\bibliography{biblio-renaming}

\appendix

\section{A lower bound on triangle-freeness}\label{sec:app}

The lower bounds in this section are a translation in the context of local certification of a lower bound obtained in \cite{DKO13}
for triangle-freeness in the congested clique model (see also \cite[Proposition 5]{CFP19} for a brief sketch).

\smallskip

We will use the following result of Ruzsa and Szemerédi \cite{RS78} (see also \cite{AES17} for a more general version, whose parameters we use below for convenience).

 \begin{lemma}[\cite{RS78, AES17}]\label{lem:aes}
For every positive integer $m$, there is tripartite graph $G$ with parts $A,B,C$, each of size $3m$, such that the edges of $G$ can
be partitioned into a family $\mathcal{F}$ of at least
$\tfrac{m^2}{\exp(20\sqrt{\log m})}$ triangles, and moreover, every triangle of $G$ lies in $\mathcal{F}$. 
\end{lemma}

In the \emph{set disjointness} problem $\mathsf{Disj}_\ell$ in the 3-party number-on-forehead model, three players Alice, Bob and Charlie each hold a subset of $[\ell]$ and they have to decide whether the intersection of the three sets is empty. The players see the subsets of the two other players, but they do not see their own subset. The \emph{deterministic communication complexity} of set disjointness in this model is the minimum number of bits the three players have to exchange before Alice can correctly decide whether the intersection of the three sets is empty (for every possible inputs). 

In the  nondeterministic setting, the players start their deterministic protocol with an additional advice string $S$, with the following constraint: if the three sets have empty intersection, then there exists an advice $S$ such that the protocol gives the correct answer, and if some element lies in the three sets, then for every advice $S$ the players give the correct answer. The \emph{nondeterministic communication complexity} minimizes the number of bits exchanged by the players \emph{plus the number of bits of the advice $S$}. We will need the following major result of Rao and Yehudayoff \cite{RY15}.

\begin{theorem}[\cite{RY15}]\label{thm:RY}
    The nondeterministic communication complexity of the set disjointness problem $\mathsf{Disj}_\ell$ in the 3-party number-on-forehead model is $\Omega(\ell)$.
\end{theorem}

We note that this result is not stated explicitly in \cite{RY15} (the result there is stated in the context of deterministic communication complexity), but their lower bound techniques also apply to the nondeterministic setting, as they bound the minimum size of a cover by 1-monochromatic cylinder intersections (see also \cite[Theorem 5.12]{RY20}). As pointed out to us by Alexander Sherstov, the result of Theorem \ref{thm:RY} can also be deduced explicitly from \cite[Theorem 5.13]{She25} by setting $f=\textsf{AND}_n$, $\epsilon=2^{-n}$, and $m$ to be a sufficiently large constant.

\medskip

We can now prove the following result.

\begin{theorem}\label{thm:trianglefree}
The local certification of triangle-freeness in  $n$-vertex graphs  requires certificates of size
$n/\exp(O(\sqrt{\log n}))$.
\end{theorem}

\begin{proof} 
Let $\ell$ be an integer. Assume that there is a local certification scheme for triangle-freeness with certificates of size $t(n)$ in $n$-vertex graphs. 
We use it to design a nondeterministic protocol for the set-disjointness problem $\textsf{Disj}_\ell$ in the 3-party number-on-forehead model.

Apply Lemma \ref{lem:aes}, and
  obtain, for $n=\sqrt{\ell}\exp(\Theta(\sqrt{\log \ell}))$ divisible by 9, a tripartite graph $F$ on $n$
  vertices, with parts $A,B,C$ of size $n/3$, and such that the edges of $F$ can
be partitioned into a family $\mathcal{F}$ of  
$\ell= \tfrac{n^2}{\exp(\Theta(\sqrt{\log n}))}$ triangles of $F$, and moreover, every triangle in $F$ lies in
$\mathcal{F}$. 

Let us denote by $\{T_i : i\in  [\ell]\}$ the $\ell$ triangles of $F$
in the family $\mathcal F$. Each edge $e$ of $F$ lies in exactly one triangle $T_i$,
and such an integer $i$ is called the \emph{index} of $e$. 

Assume Alice, Bob and Charlie are given subsets $X_A,X_B,X_C\subseteq [\ell]$. Their goal is to decide whether their three sets have empty intersection. Let $X=(X_A,X_B,X_C)\subseteq [\ell]^3$, and  define the spanning subgraph $G_X$ of $F$, where:
\begin{itemize}
    \item an edge of $F$ between $A$ and $B$ lies in $G_X$ if and only if its index lies in $X_C$;
    \item an edge of $F$ between $B$ and $C$ lies in $G_X$ if and only if its index lies in $X_A$;
    \item an edge of $F$ between $A$ and $C$ lies in $G_X$ if and only if its index lies in $X_B$.
\end{itemize}

Note that a triangle $T_i$ of $F$ survives in $G_{X}$ if and only if
$i \in X_A\cap X_B\cap X_C$, and thus $G_{X}$ is triangle-free 
if and only if $X_A\cap X_B\cap X_C= \emptyset$.

\smallskip

Consider a local certification scheme for triangle-freeness in $G_X$ with certificates of size $t(n)$. Then the total number of bits of
certificates assigned to the vertices 
in $G_{X}$ is $n\cdot t(n)$. The certificates of all the vertices of $G_X$ are given to Alice, Bob, and Charlie as advice. 

Alice simulates the vertices of $A$, Bob simulates the vertices of $B$, and Charlie simulates the vertices of $C$. Since each player sees the sets of $X$ associated to the other two players, they know exactly which edges of $F$ incident to their vertex set exist in $G_X$. In particular, given the advice, each player can simulate the local verification process at each of their vertices, and return that the three sets $X_A,X_B,X_C$ have empty intersection if and only if all the vertices in their vertex set ($A$, $B$, or $C$) declare that $G_X$ is triangle-free. 
Bob and Charlie can then each send a single bit to Alice, giving their answer, and Alice can then correctly decide whether   $X_A,X_B,X_C$ have empty intersection.
This shows the correctness of the nondeterministic protocol for $\textsf{Disj}_\ell$ in the 3-party number-on-forehead model, given a local certification scheme for triangle-freeness in $n$-vertex graphs with certificates of size $t(n)$.

The complexity of the protocol is $n\cdot t(n)+2=\Omega(n \cdot t(n))$.  By Theorem \ref{thm:RY}, there is an absolute constant $c>0$ such that $n\cdot t(n)\ge c \ell$ and thus \[t(n)\ge c \ell/n=n/\exp(O(\sqrt{\log n})),\] as desired.
\end{proof}

By considering the complement of $F$ and $G_X$ instead of $F$ and $G_X$ in the proof above, we immediately obtain the following analogous result for anti-triangle-freeness.

\begin{theorem}\label{thm:antitrianglefree}
The local certification of anti-triangle-freeness in  $n$-vertex graphs  requires certificates of size
$n/\exp(O(\sqrt{\log n}))$.
\end{theorem}

\end{document}